\documentclass[10pt]{article}
\pdfoutput=1
\usepackage{amsmath}
\usepackage{amsfonts}
\usepackage{amssymb}
\usepackage{amsthm}
\usepackage{amscd}
\usepackage{natbib}
\usepackage{graphicx}
\usepackage{verbatim}
\usepackage[pdftex,
	bookmarks=false,
	colorlinks,
	pdftitle={Minimum multicuts and Steiner forests for Okamura-Seymour graphs},
	pdfauthor={Arindam Pal},
	pdfsubject={Minimum multicuts for Okamura-Seymour graphs},
	pdfkeywords={graph algorithms, planar graphs, minimum multicuts, Steiner forests, Okamura-Seymour instance},
	pdfstartview=FitH,
	pdfdisplaydoctitle=true
]{hyperref}

\newtheorem{theorem}{Theorem}
\newtheorem{lemma}[theorem]{Lemma}
\newtheorem{corollary}[theorem]{Corollary}

\def\NN{\mathbb{N}}

\begin{document}
\title{\textsc{Minimum multicuts and Steiner forests for Okamura-Seymour graphs}}
\author{Arindam Pal\\\\
\textsl{Department of Computer Science and Engineering}\\
\textsl{Indian Institute of Technology, Delhi}\\
\textsl{New Delhi -- 110016, India}\\
\texttt{\href{mailto:arindamp@cse.iitd.ernet.in}{arindamp@cse.iitd.ernet.in}}}
\date{February 27, 2011}
\maketitle

\begin{abstract}
We study the problem of finding \emph{minimum multicuts} for an undirected \emph{planar graph}, where all the terminal vertices are on the boundary of the outer face. This is known as an \emph{Okamura-Seymour instance}. We show that for such an instance, the minimum multicut problem can be reduced to the \emph{minimum-cost Steiner forest} problem on a suitably defined \emph{dual graph}. The minimum-cost Steiner forest problem has a 2-approximation algorithm. Hence, the minimum multicut problem has a 2-approximation algorithm for an Okamura-Seymour instance.
\end{abstract}
\section{Introduction}
Computing the \emph{minimum multicut} of a graph is an important problem in combinatorial optimization. The problem is formally defined below.\\\\
\textsc{Minimum Multicut:} Given an undirected graph $G = (V,E)$ on $n$ vertices and $m$ edges with edge costs $c_e$ and a set $T = \{(s_1,t_1),\ldots,(s_k,t_k)\}$ of $k$ terminal pairs. Our goal is to find a set of edges $F$ of minimum cost separating all the terminal pairs $(s_i,t_i), 1 \le i \le k$, \emph{i.e.}, in the subgraph $H = (V,E \setminus F)$, there is no path between $s_i$ and $t_i$, for $1 \le i \le k$.\\\\
The minimum multicut problem is \textsc{NP-hard} and \textsc{APX-hard}, even for trees, which excludes the possibility of a PTAS. Garg, Vazirani and Yannakakis  gave a 2-approximation algorithm for trees \citep{GargVY97} and an $O(\log k)$-approximation algorithm for general undirected graphs \citep{GargVY96}. When $G$ is planar and all the terminals are on the boundary of the outer face of $G$, we denote such an instance of the minimum multicut problem an \emph{Okamura-Seymour} instance. The problem of finding edge-disjoint paths on such type of graphs were studied by Okamura and Seymour \citep{OS81}. They showed the following important theorem ($d_F(S)$ denotes the number of edges in $F$ exactly one of whose endpoints is in $S$).

\begin{theorem}[\textsc{Okamura-Seymour}]
Let $G=(V,E)$ be an undirected planar graph and let $R=\{(s_i,t_i): s_i,t_i \in V, 1 \le i \le k\}$ be a set of terminal pairs. Suppose the following conditions are satisfied.
\begin{enumerate}
	\item All terminals are on the boundary of the outer face of $G$.
	\item The Euler condition is satisfied: $(V,E \cup R)$ is Eulerian.
	\item The cut condition is satisfied: $d_E(S) \ge d_R(S)$, for all $S \subseteq V$.
\end{enumerate}
Then there exist edge-disjoint paths between $s_i$ and $t_i$, for $1 \le i \le k$.
\end{theorem}

\noindent They also showed the following corollary about the associated \emph{multicommodity flow} problem.

\begin{corollary}[\textsc{Okamura-Seymour}]
Let $G=(V,E)$ be an undirected planar graph and let $R=\{(s_i,t_i): s_i,t_i \in V, 1 \le i \le k\}$ be a set of terminal pairs. Suppose the following conditions are satisfied.
\begin{enumerate}
	\item All terminals are on the boundary of the outer face of $G$.
	\item The cut condition is satisfied: $d_E(S) \ge d_R(S)$, for all $S \subseteq V$.
\end{enumerate}
Then there exists a feasible multicommodity flow between $s_i$ and $t_i$, for $1 \le i \le k$. Moreover, if $c_e \in \NN, \forall e \in E$ and $d_i \in \NN, \forall i: 1 \le i \le k$, then there exists a half-integer multicommodity flow.
\end{corollary}

Schwarzler proved that computing edge-disjoint paths in such graphs \emph{without the Euler condition} is \textsc{NP-hard} \citep{Schwarzler09}. Wagner and Weihe gave a linear-time algorithm to compute edge-disjoint paths in such graphs \citep{WagnerW95}. The multicommodity flow problem for an Okamura-Seymour instance was studied by Matsumoto et al \citep{MatsumotoNS85}. Their algorithm decides whether $G$ has a feasible multicommodity flow, each from a source to a sink and of a given demand, and actually finds them if $G$ has one. If $G$ has $n$ vertices and $k$ source-sink pairs, their algorithm takes $O(kn + n^2\sqrt{\log n})$ time and $O(kn)$ space. However, the dual problem of computing the minimum multicut has not been addressed by them. We take a step in that direction. Our main result is the following.

\begin{theorem}
The minimum multicut problem on an Okamura-Seymour instance can be reduced to the minimum-cost Steiner forest problem on an appropriately defined dual graph. The minimum-cost Steiner forest problem has a $2$-approximation algorithm. Hence, the minimum multicut problem has a $2$-approximation algorithm for an Okamura-Seymour instance.
\end{theorem}

\noindent The minimum-cost Steiner forest problem is formally defined below.\\\\
\textsc{Minimum-cost Steiner forest:} Given an undirected graph $G = (V,E)$ on $n$ vertices and $m$ edges with edge costs $c_e$ and a set $T = \{(s_1,t_1),\ldots,(s_k,t_k)\}$ of $k$ terminal pairs. Our goal is to find a set of edges $F$ of minimum cost connecting all the terminal pairs $(s_i,t_i), 1 \le i \le k$, \emph{i.e.}, in the subgraph $H = (V,F)$, there is some path between $s_i$ and $t_i$ for $1 \le i \le k$.\\\\
There is a polynomial-time 2-approximation algorithm for the minimum-cost Steiner forest problem in any undirected graph, due to Goemans and Williamson \citep{GoemansW95}. A comprehensive survey of all these results is given in \citep{Frank90}. For a more recent survey, see \citep{NavesS09}.

\section{Reduction of minimum multicut to Steiner forest}
\noindent Let $OF$ be the outer face of $G$ and $B$ be the boundary of $OF$. We construct the dual graph $G_d = (V_d, E_d)$ of the planar graph $G$ as follows.
\begin{enumerate}
	\item For each finite face $f$ of $G$, we associate a dual vertex $v_f$ of $G_d$.
	\item For each edge $e = (u,v) \in E$ on the boundary of two finite faces $f$ and $g$ of $G$, we associate a dual edge $e_d = (v_f,v_g) \in E_d$.
	\item For each terminal pair $(s_i,t_i)$, we add two dual vertices $u_i,v_i \in V_d$ such that $u_i,v_i$ are on $OF$ in the planar embedding of $G$. The terminal pair $(s_i,t_i)$ divides $B$ into two parts. Let us denote by $[s_i,t_i]$ the portion encountered while traversing $B$ from $s_i$ to $t_i$ in the anti-clockwise direction, and by $[t_i,s_i]$ the other portion. We associate $u_i$ with $[s_i,t_i]$ and $v_i$ with $[t_i,s_i]$.
	\item For each edge $e \in E$ on $B$, we add a dual vertex $v_e$ on $OF$ and a dual edge $e_d = (v_e,v_f) \in E_d$, where $e$ is the common boundary of $OF$ and a finite face $f$.
	\item For each edge $e \in E$ on $B$ that we come across during the traversal of $[s_i,t_i]$, we add a dual edge $e_d = (u_i,v_e) \in E_d$, where $e$ is the common boundary of $OF$ and a finite face $f$ and $v_e$ is a dual vertex on $OF$ added in the last step. Similarly, we add dual edges $(v_i,v_e)$ for the portion $[t_i,s_i]$.
\end{enumerate}
We assign a cost $c_e$ to dual edges of type (2) and (4). For dual edges of type (5), we assign a cost $N = \sum_{e \in E} c_e$, \emph{i.e.}, the sum of all edge costs in $G$. Note that there is a technical difference between the dual graph defined here and the traditional definition. Here we have many dual vertices on the infinite face, whereas traditionally there is only one dual vertex on the infinite face. This is illustrated in Figure \ref{fig1}. In general, the dual graph $G_d$ is not planar. We say that a set of edges $F \subseteq E$ \emph{corresponds} to a set of edges $F_d \subseteq E_d$, if for every edge $f \in F$, there exists one and only one edge $f_d \in F_d$ such that $f_d$ is the dual edge of $f$.

\begin{figure}
\centering
\includegraphics[width=0.7\textwidth]{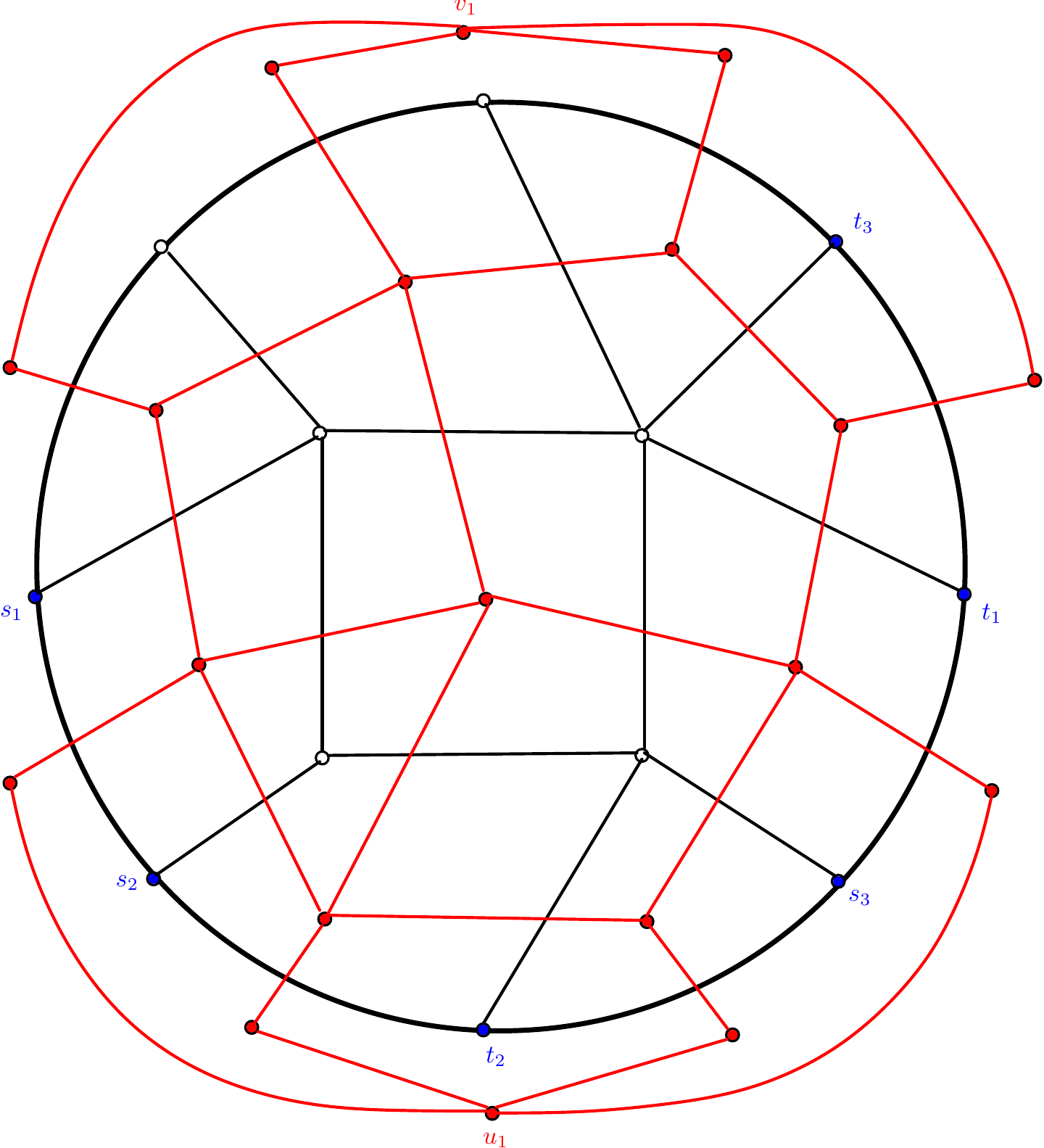}
\caption{A planar graph $G$ and its dual $G_d$. Terminals are in blue, dual vertices and edges are in red. Only $u_1$ and $v_1$ are shown for clarity. The other vertices $u_2,v_2,u_3,v_3$ and the edges incident on them can be drawn similarly.}
\label{fig1}
\end{figure}

Let $R = \{u_1,v_1,\ldots,u_k,v_k\}$ be the set of dual vertices on $OF$ containing all the $(u_i,v_i)$ pairs and let $S = V_d - R$ be the set of remaining vertices of $G_d$. We call $R$ the set of \emph{required vertices} and $S$ the set of \emph{Steiner vertices}. Our goal is to connect all the pairs $(u_i,v_i) \in R$. Let $SF$ be a \emph{Steiner forest} of the dual graph $G_d$ connecting all $(u_i,v_i)$ pairs in $R$ using some Steiner vertices in $S$. We call a dual edge an \emph{internal edge}, if both its endpoints are inside the outer boundary $B$, an \emph{external edge}, if both its endpoints are outside $B$, and a \emph{crossing edge}, if one endpoint is inside and the other is outside $B$. Note that the external edges have a cost of $N$ and the other edges have a cost of the corresponding primal edge. For convenience, we prove the following technical result.

\begin{lemma}
Suppose $P_i$ is a path in $G$ between the terminals $s_i$ and $t_i$. Further, $Q_i$ is a dual path in $G_d$ between the dual vertices $u_i$ and $v_i$ such that only its first and last edges are external edges. Then, $P_i$ and $Q_i$ must cross each other. Hence, if $(V(G),E(G)\setminus F)$ contains no $s_i$-$t_i$-path, then the dual edge set of $F$ plus two external edges contains a $u_i$-$v_i$-path.
\end{lemma}
\begin{proof}
The dual path $Q_i$ starts and ends at the outer face $OF$ and passes through the interior of $B$. Hence, it separates $s_i$ and $t_i$, \emph{i.e.} $s_i$ and $t_i$ lie on different sides of $Q_i$. So, any path $P_i$ between the terminals $s_i$ and $t_i$ must intersect the dual path $Q_i$.
\end{proof}

We claim that the minimum-cost Steiner forest can't use too many external edges of cost $N$.

\begin{lemma}
Let $MSF$ be a minimum-cost Steiner forest in $G_d$ connecting the pairs $(u_i,v_i) \in R$. Then, $MSF$ can use at most $2k$ external edges, one for each required vertex $w \in R$.
\end{lemma}
\begin{proof}
Suppose this is not the case. If we use more than 2 external edges to connect a pair $(u_i,v_i) \in R$, we need at least 4 such edges. This follows from a simple pairing argument: if we go outside, we have to come back inside using an external edge and if we use only external edges, we need at least 4 such edges to connect $u_i$ and $v_i$. But we need at least 2 external edges to connect any pair, because all the edges incident on any $u_i,v_i$ are external edges. Hence, cost of $MSF$ is at least $(2k+2)N$, whereas a Steiner forest using $2k$ external edges will have cost at most $2kN + \sum_{e \in E} c_e < (2k+2)N$, which is strictly less than the cost of the $MSF$, contradicting the fact that $MSF$ is a minimum Steiner forest.
\end{proof}

Next we show the relationship between Steiner forests and multicuts.

\begin{lemma}
Every Steiner forest $SF$ in $G_d$ corresponds to a multicut $MC$ in $G$.
\end{lemma}
\begin{proof}
Consider the dual path $Q_i$ in $SF$ between the dual vertices $u_i$ and $v_i$. By Lemma 2, any path $P_i$ in $G$ between the terminals $s_i$ and $t_i$ must intersect $Q_i$. Therefore, primal edges corresponding to $Q_i$ is a $s_i-t_i$ cut. Since this is true for all terminal pairs, the primal edges corresponding to $SF$ is a multicut $MC$ in $G$.
\end{proof}

We observe that not every multicut $MC$ in $G$ is a Steiner forest $SF$ in $G_d$. We say that $MC$ is a \emph{minimal multicut} if for any edge $e \in MC, MC - \{e\}$ is not a multicut, \emph{i.e.} there exists terminals $s_i$ and $t_i$ such that there is a path $P_i$ between $s_i$ and $t_i$ in $G - MC \cup \{e\}$. We show the relationship between minimal multicuts and Steiner forests in the next lemma.

\begin{lemma}
Every minimal multicut $MC$ in $G$ corresponds to a Steiner forest $SF$ in $G_d$.
\end{lemma}
\begin{proof}
Since $MC$ is a multicut in $G$, there is no path between $s_i$ and $t_i$ in $G - MC$. By Lemma 2, there must be a dual path between $u_i$ and $v_i$. Hence all pairs $(u_i,v_i) \in R$ are connected. Moreover, since $MC$ is minimal, there is only one path between $u_i$ and $v_i$. Thus, the dual edges corresponding to $MC$ is a Steiner forest $SF$ in $G_d$.
\end{proof}

Combining Lemmas 3, 4 and 5, we arrive at Theorem 3.

\section{Conclusion and open problems}
In this paper, we studied the minimum multicut problem for an undirected planar graph, where all the terminal vertices are on the boundary of the outer face. We showed its relation to the minimum-cost Steiner forest problem in the dual graph and gave a $2$-approximation algorithm. Are there similar relationships between these problems in a general undirected graph? Is there a direct $2$-approximation algorithm for the minimum multicut problem without reducing it to the Steiner forest problem?

\bibliographystyle{plainnat}
\bibliography{mcos}

\end{document}